\documentclass[12pt]{article}
\newtheorem{theorem}{Theorem}
\newtheorem{lemma}{Lemma}
\newtheorem{corollary}{Corollary}

\newtheorem{definition}{Definition}
\newtheorem{example}{Example}

\newcommand{\blackslug}{\mbox{\hskip 1pt \vrule width 4pt height 8pt
depth 1.5pt \hskip 1pt}}
\newcommand{\qed}{\quad\blackslug\lower 8.5pt\null\par\noindent}
\newenvironment{proof}{\par\noindent{\bf Proof:}}{\qed \par}

\newcommand{\cR}{\mbox{${\cal R}$}}
\title{Quality of local equilibria in discrete exchange economies}
\author{Daniel Lehmann \\ School of Computer Science and Engineering, \\
Hebrew University, Jerusalem 91904, Israel \\
lehmann@cs.huji.ac.il, Tel: +972544489124 \\
}
\date{February 2020}

\begin{document}
\maketitle

%%Towards v2 in arXiv: 
% - finding a local optimum in a submodular economy is hard by a result of Babichenko-Dobzinski-Nisan
% - For XOS valuations (not only submodular) the nobuy property for single items imply
% the nobuy property for arbitrary bundles (??)

\begin{abstract}
This paper defines the notion of a local equilibrium of quality $(r , s)$, 
\mbox{$0 \leq r , s$},
in a discrete exchange economy: a partial allocation and item prices that guarantee
certain stability properties parametrized by the numbers $r$ and $s$.
The quality $( r , s )$ measures the fit between the allocation and the prices:
the larger $r$ and $s$ the closer the fit.
For \mbox{$r , s \leq 1$} this notion provides a graceful degradation 
for the conditional equilibria of ~\cite{Fu_Kleinberg_Lavi:EC12} which are exactly the
local equilibria of quality $( 1 , 1 )$.
For \mbox{$1 < r , s $} the local equilibria of quality $( r , s )$ are {\em more stable} 
than conditional equilibria.
Any local equilibrium of quality $( r , s )$ provides, 
without any assumption on the type of the agents' valuations, an allocation
whose value is at least \mbox{$\frac{r s} { 1 + r s }$} the optimal fractional allocation.
In any economy in which all agents' valuations are $a$-submodular, 
i.e., exhibit complementarity bounded by \mbox{$a \: \geq \: 1$},
there is a local equilibrium of quality $( \frac{1} {a} , \frac{1}{a} )$. 
In such an economy any greedy allocation provides a local equilibrium of quality 
\mbox{$( 1 , \frac{1}{a} ) $}.
Walrasian equilibria are not amenable to such graceful degradation.

Keywords: discrete exchange economies; local equilibria; item prices

\end{abstract}
 
\section{Background} \label{sec:background}
Economic theory, since Adam Smith~\cite{AdamSmith:Wealth}, has a Leibnizian flavor: 
some ``invisible hand" provides an efficient situation.
It also has a distributed flavor: there is no central command.
The means by which such a feat is achieved are prices:
publicly posted prices are accepted by the agents and this enables them to sell
and buy to enhance their welfare without negotiating with other agents and
such sales and purchases lead to a situation that is favorable to everybody.
L\'{e}on Walras~\cite{Walras:1874} proposed a formal description of the situation
and the prices obtained, now called a competitive or Walrasian equilibrium.
Walrasian prices equate supply and demand.
Once such prices are publicly known and no one can influence them, if each agent pursues
his or her own individual interest without any consideration for others, each agent will obtain
exactly all he or she wants.
This stands in sharp contrast with other social situations in which, for the benefit of all, 
each one has to give up on some of his or her wishes.
This Leibnizian flavor of economic theory has a pervasive influence on public opinion and economic policy.

It was only Abraham Wald~\cite{Wald:1936} who brought
to the attention of economists the problem of proving rigorously
the existence of such equilibria.
Intuitively, the existence of competitive equilibria, or their enforcement in a
market, depends on each trader renouncing the idea of influencing 
the prices and on the prices equating supply and demand. 
This typically happens in markets with a large number of traders.

For economies of divisible objects, the existence of a competitive equilibrium has
been proved in two different types of situations.
Arrow and Debreu~\cite{ArrowDebreu:54} proved that, if the valuation of every trader is
 concave, then the existence of a competitive equilibrium is guaranteed. 
It is clear that real economic agents do not always exhibit concave valuations: 
one may well hesitate between a beef roast for $x$ shekels and a lamb shoulder 
for $y$ shekels but not be interested in buying half the roast and half the shoulder for 
\mbox{$\frac{x + y }{2}$} shekels.
Their result holds for any number of traders.
Aumann~\cite{Aumann_Continuum:66} proves the existence of a competitive equilibrium for 
any valuations of the individual traders, but only if there is a
continuum of traders. This is also an idealization of real markets.

For economies of indivisible objects,
Kelso and Crawford~\cite{KelsoCraw:82} proved, for any number of agents, the existence
of a competitive equilibrium under a different assumption: they assume
that the valuation of each trader is (gross) substitutes.
In~\cite{LLN:GEB} it was shown that substitutes valuations have zero measure among
all valuations and therefore, given any substitutes valuation, 
some arbitrarily small modification will define a valuation that is {\em not} substitutes. 
The {\em substitutes} assumption is therefore very restrictive.
Larger families of valuations for which a Walrasian equilibrium always exists 
are described in~\cite{Sun_Yang:econometrica, Hatfield_Milgrom:2005, Hatfield+4:JEP}.
Each of them assumes some additional structure on the set of items, in addition to some 
conditions on the traders valuations.

If we consider now real economic activity, Walrasian equilibria can be considered a suitable 
justification of Adam Smith's general perspective only if real markets can be accurately 
described by formal economies that possess a Walrasian equilibrium. 
The results mentioned above, and others, show that, in formal economies, even those that
intend to model real agents, no Walrasian equilibrium is guaranteed.
One of the greatest achievements of the notion of a Walrasian equilibrium is that 
it explains prices, which are ubiquitous in markets. 
But prices seem to be there even in the absence of a Walrasian equilibrium and, therefore,
cannot always be explained by such equilibria.
We ask what drives such prices and what is their function?

One could expect that markets, if they do not possess a Walrasian equilibrium, possess
something close to a Walrasian equilibrium that would define prices.
But, so far, no such satisfactory notion has been proposed. 
Appendix~\ref{sec:quasi-Walras} shows that there is little hope that such a notion exists.
But, this paper shows that, on the contrary, the notion of a conditional equilibrium, 
defined  in~\cite{Fu_Kleinberg_Lavi:EC12}, can be parametrized to provide a graceful
degradation that provides approximate equilibria to essentially all discrete economies. 

\section{This paper} \label{sec:this}
This paper proposes a perspective change:
given an exchange economy, don't ask whether it possesses a Walrasian equilibrium,
ask how good is the best equilibrium.
To this purpose the notion of an equilibrium of quality \mbox{$( r , s )$}, 
for \mbox{$0 \leq r , s$} will be defined in Definition~\ref{def:local_eq}.
An equilibrium consists of a partial allocation and a price vector 
and the quantities $r$ and $s$ measure the tightness of the fit between allocation 
and price vector.
The notion of an equilibrium developed in this work, termed a {\em local equilibrium}, 
generalizes the notion of a {\em conditional equilibrium} 
defined in~\cite{Fu_Kleinberg_Lavi:EC12}.
A local equilibrium of quality \mbox{$( 1 , 1 )$} is exactly a conditional equilibrium.

One may have preferred to consider a notion of equilibrium that generalizes that of
a Walrasian equilibrium, not that of a conditional (local) equilibrium.
No such notion has been shown, so far, to have interesting properties.
Appendix~\ref{sec:quasi-Walras} defines such a possible notion of equilibrium, parametrized
by a single parameter $q$,and shows that, in even simple exchange economies, 
only equilibria of quality $0$ may exist, 
in stark contradiction with  Theorems~\ref{the:equilibrium} and~\ref{the:greedy} below.

Many recent works, e.g., \cite{Fu_Kleinberg_Lavi:EC12, Christodoulou_Kovacs_Schapira:JACM, 
Feldman_Gravin_Lucier:posted, Babaioff_Endowment:EC18} 
have studied the problem of approximating the social optimum, 
the role that prices can play in doing so and the properties of different notions of equilibria.
The general impression that one gathers is that as long as the agents' valuations are 
submodular, things are reasonably well understood and one can justify the general perspective
portrayed by Adam Smith, by generalizing Walras' model.
But, on the whole, such approaches have not been convincing 
so far when the agents' valuations exhibit complementarities, as is often 
the case for real life agents.
Section 7 of~\cite{LLN:GEB} introduced the notion of an $a$-submodular valuation,
i.e., a valuation exhibiting complementarity bounded by $a$, \mbox{$1 \leq a $}.
The parameter $a$ measures how far a valuation is from being submodular.
This paper shows that most results obtained about exchange economies of submodular agents
can be extended to such economies of $a$-submodular agents: the strength of the result 
obtained depends gracefully on the parameter $a$.
Since essentially any valuation is $a$-submodular for some $a$, this paper generalizes
what is known for exchanges of submodular agents to almost arbitrary agents.

This paper proposes an original justification for Adam Smith's general perspective. 
Each agent has an initial endowment.
Agents perform simple profitable bilateral trades: 
two agents agree that a single item will be transferred
from one agent to the other for a certain amount of money.
Such trades have a double effect: first the social welfare is increased and secondly 
item prices become more and more publicly known.
After a certain time we expect to find the economy in a state where the social
welfare cannot be improved upon by transfers of a single item from an agent to another
(we shall call such a situation a local optimum) and where a price is publicly known
for each item.
Such prices support the allocation obtained in a way to be described below
that we shall call a local equilibrium. 
The quality \mbox{$( r , s )$} of a local equilibrium measures the fit between the allocation and the prices.
We shall show that any local optimum, i.e., any allocation in which no simple bilateral trade 
can be profitable to both the seller and the buyer, provides a local equilibrium whose quality
depends on the amount of complementarity exhibited by the agents' valuations.
The less complementarity, the better the quality.
We shall also show that any local equilibrium of high quality has a high social value, 
i.e., its value is close to the socially optimal value.
The general perspective becomes: the agents trade bilaterally and in doing so item prices
crystallize and all this leads to a situation in which no simple bilateral trade 
can be profitable for both agents and to prices making the situation a local equilibrium.
In typical situations, we claim, the amount of complementarity is limited, 
the quality of the local equilibrium obtained is high and therefore its
social value is close to optimal, but not optimal.
At the end the agents are not allocated their preferred bundle at the posted prices 
but a bundle that cannot be improved upon by any simple bilateral trade 
at the posted prices.
If one wants to design a market in which an increased social welfare is attained, 
one should design means to support trades more complex than transfers of
single items, e.g., bilateral trading of bundles or multilateral trades.

\section{Plan of this paper} \label{sec:plan}
Section~\ref{sec:exchange} defines discrete exchange economies, fractional allocations
and the fractional optimal allocation.
Section~\ref{sec:Walras} briefly recalls Walrasian equilibria and points to an original
presentation of the main results about them in an Appendix.
In Section~\ref{sec:local_eq} we define a notion of equilibrium, 
 a {\em \mbox{$ ( r , s ) $}-local equilibrium}, parametrized by two quality parameters $r$,
and $s$ such that \mbox{$0 \: \leq r , s $}.
If \mbox{$r , s \leq 1$} this notion is weaker than a Walrasian equilibrium.
In Section~\ref{sec:first} we show that $( r , s )$-local equilibria are  
\mbox{$\frac{r s}{1 + r s }$}-efficient.
The remainder of this paper is devoted to the defense of the following thesis:
typical economies possess high quality local equilibria and bilateral trading can discover
them or move the economy  towards them.
In Section~\ref{sec:local_existence} we describe and discuss examples of
economies illustrating the absence of high quality local equilibria.
Section~\ref{sec:bounded} recalls the bounded complementarity introduced 
in~\cite{LLN:GEB}: $a$-submodular valuations are valuations whose degree 
of complementarity is bounded by the parameter $a$, \mbox{$1 \leq a$}.
Almost all valuations are $a$-submodular for some $a$.
It deepens their study.
Section~\ref{sec:opt_equil} provides a characterization of \mbox{$( r , s )$}-local equilibria in 
$a$-submodular economies.
Section~\ref{sec:local_optimum} recalls the notion of a local optimum 
from~\cite{Babaioff_Endowment:EC18}. 
It proves a second welfare theorem: in an $a$-submodular economy, 
any local optimum can be associated with a price vector to provide 
a \mbox{$( \frac{1}{a} , \frac{1}{a} ) $}-local equilibrium.
Therefore, in an $a$-submodular economy every local optimum provides a high-quality
local equilibrium.
In particular, if all valuations are submodular any local optimum can be associated 
with a \mbox{$ ( 1 , 1 ) $}-local equilibrium.
Section~\ref{sec:greedy} sharpens the results of~\cite{LLN:GEB} about greeedy allocations
in the presence of bounded complementarity.
In such situations any greedy allocation method provides a local equilibrium and a good
approximation of the social optimum.
Section~\ref{sec:substitutes} deals with the case all agents' valuations are substitutes.
In this case Walrasian equilibria are exactly those \mbox{$ ( 1 , 1 ) $}-local equilibria 
that satisfy an additional condition: no agent is interested in exchanging one of his items 
for an item he does not possess at the given prices.
Section~\ref{sec:open} presents a list of open questions.
Section~\ref{sec:conclusion} concludes with a reflexion on the role item prices 
play in exchange economies.

\section{Exchange economies, allocations, fractional allocations and prices} 
\label{sec:exchange}
We consider exchange economies of indivisible objects, with
private values, quasi-linear utilities,
no externalities, free disposal and normalization.

\begin{definition} \label{def:exchange}
An exchange economy is defined by a finite set of {\em indivisible} objects $X$,
of size $m$, a finite set $N$ of agents, of size $n$, and by $n$ valuations:
the valuation \mbox{$v_{i} : {2}^{X} \longrightarrow \cR$}, \mbox{$i \in N$} 
describes the preferences of agent $i$.
A bundle \mbox{$D \subseteq X$} possesses the value  
\mbox{$v_{i}(D)$} for agent \mbox{$i \in N$}.
We shall assume that those valuations satisfy:
\begin{itemize}
\item
{\bf Free disposal:} $v_i(A) \le v_i(B)$, whenever $A \subseteq B$.
\item
{\bf Normalization:} $v_i(\emptyset)=0$.
\end{itemize}
\end{definition}

In an exchange economy a partial allocation (hereafter called allocation) is a function 
\mbox{$f : X \longrightarrow N \cup \{ unallocated \}$}.
Item $j$ of $X$ is allocated to agent $f(j)$ or left unallocated.
An allocation is {\em total} if no item is left unallocated.
The set of items allocated to agent $i$, $f^{- 1}(i)$ in allocation $f$ 
will be denoted by $S_{i}^{f}$.
For an allocation $f$, we define its social value by:
\mbox{$val(f) = \sum_{i \in N} v_{i}(S^{f}_{i})$}.
In a given exchange economy, the social value of the allocation of maximal social value is
denoted by $M$.
We shall follow the convention that is now well established in complexity theory
and say that an allocation $f$ is an $a$-approximation (\mbox{$a \geq 1$}) 
of the social optimum iff \mbox{$ M \, \leq \, a \, val(f) $}.

In~\cite{BikhMamer:Equ} Bikhchandani and Mamer considered {\em fractional allocations}, 
a generalization of the notion of an allocation .
A fractional allocation consists of a nonnegative number $x_{i}^{D}$ for every \mbox{$i \in N$}
and every \mbox{$D \subseteq X$} satisfying the constraints:
\begin{enumerate}
\item 
for any \mbox{$i \in N$}, 
\begin{equation} \label{eq:con1LPR}
\sum_{D \subseteq X} \: x_{i}^{D} \: \leq \: 1,
\end{equation}
and
\item for any \mbox{$j \in X$}, 
\begin{equation} \label{eq:con2LPR}
\sum_{i \in N} \sum_{D \subseteq X , j \in D} \: x_{i}^{D} \: \leq \: 1.
\end{equation}
\end{enumerate}
The value of a fractional allocation $x$ is defined by:
\[
val(x) \: = \: \sum_{i \in N} \sum_{D \subseteq X} x_{i}^{D} v_{i}(D).
\]
The value of the fractional allocation of maximal value will be denoted $M_{F}$.
It is the solution of the linear program \textbf{LP}.
\begin{center}
\textbf{Linear Programming (LP):}
\end{center}
\nopagebreak Maximize
\begin{equation}
\label{eq:objLPR}
\sum_{i \in N} \sum_{D \subseteq X} x^{i}_{D} v_{i}(D)
\end{equation}
under the constraints
\begin{equation}
\label{eq:2con1LPR}
\sum_{D \subseteq X , j \in D} \sum_{i \in N} x^{i}_{D} \leq 1 , {\rm \ for \ all \ }
j \in X,
\end{equation}
\begin{equation}
\label{eq:2con2LPR}
\sum_{D \subseteq X} x^{i}_{D} \leq 1
{\rm \ for \ all \ } i \in N , {\rm \ and\ }
\end{equation}
\begin{equation}
\label{eq:intconsLPR}
x^{i}_{D} \in [0 , 1] , {\rm \ for \ all \ } i \in N , D \subseteq X.
\end{equation}
The optimal solution set to \textbf{LP} can be viewed as the set of
efficient {\em fractional} allocations.
Agents may be allocated fractional bundles of the form:
\mbox{$\alpha \in [0,1]^{2^{X}}$} as long as
\mbox{$\sum_{D \subseteq X} \alpha(D) \leq 1$} 
(this is constraint~(\ref{eq:2con2LPR}) ) and, for each item,
the sum of the fractions of it that are allocated does not exceed one
(this is constraint~(\ref{eq:2con1LPR}) ).
Agent $i$ values the fractional bundle $\alpha$ at:
\mbox{$\sum_{D \subseteq X} \alpha(D) v_{i}(D)$}.

Any allocation $f$ is the fractional allocation for which 
\mbox{$x_{i}^{D} = 1$} iff \mbox{$D = S_{i}^{f}$} and \mbox{$x_{i}^{D} = 0$}
otherwise. Clearly, in any exchange economy, $M$ is the solution of the integer version of 
\textbf{LP} where \mbox{$x^{i}_{D} \in \{0 , 1\}$}. 
Therefore we have \mbox{$M \: \leq \: M_{F}$}.

A price vector is a function \mbox{$p : X \longrightarrow \cR_{+}$} that
assigns a nonnegative real number (its price) to each item.
The price of item $j$, $p(j)$ will be denoted $p_{j}$.

\section{Walrasian equilibria} \label{sec:Walras}
In preparation for Section~\ref{sec:local_eq} where another notion of equilibrium
will be defined, the reader can find in  Appendix~\ref{sec:Walras_app}
the definition and the properties of the classical notion of a Walrasian equilibrium.
The presentation there is original.

In a Walrasian equilibrium every agent is allocated the bundle he prefers amongst all
possible bundles, if only he considers he cannot have any influence on the prices.
A Walrasian equilibrium is the best of all possible situations for each and every agent,
at the publicly posted prices.

The existence of a Walrasian equilibrium is not guaranteed in general.
It is only in exchange economies in which every agent has 
a {\em (gross) substitutes} valuation
that such an equilibrium is guaranteed to exist, by a result of~\cite{KelsoCraw:82}.
By~\cite{GulStacc:99}, to any valuation that is not substitutes one may add 
unit-demand valuations to define an exchange economy without a Walrasian equilibrium.
The family of substitutes valuations has zero measure, as shown in~\cite{LLN:GEB},
which implies that any substitutes valuation can be approached as close as one wants by valuations
that are {\em not} substitutes.
Therefore one can say that Walrasian equilibria are quite rare.

This paper's goal is to propose a less optimistic but more realistic view of the states
into which economies can evolve.
Agents will not find themselves in the best of all possible worlds but in a relatively good
situation, a situation that cannot be improved upon easily.

\section{Local equilibria} \label{sec:local_eq}
We shall now define the central  notion of this paper. 
A local equilibrium comprises a partial allocation, a price vector and two quality parameters.
The quality parameters are real numbers \mbox{$r , s \geq 0$}: 
they measure the fit between the allocation and the prices, 
{\em not the social value of the allocation}.
A local equilibrium of quality, \mbox{$ ( 1 , 1 ) $}, is exactly 
a conditional equilibrium as defined  in~\cite{Fu_Kleinberg_Lavi:EC12}.

In a conditional equilibrium two conditions are satisfied.
\begin{enumerate}
\item the allocation gives every agent a bundle with a nonnegative utility, at the given prices. 
In other terms no agent is willing to give his whole bundle back, given the prices,
but he could, for example, prefer selling a subset of his bundle or exchanging an item $k$ 
he has been allocated for an item $l$ allocated to some other agent and pay 
\mbox{$p_{l} - p_{k}$},
\item no agent wishes to buy any set of items he does not own, at the given prices.
\end{enumerate}

We generalize this definition by adding quality parameters $r$, $s$.
This paper is mainly concerned with the case \mbox{$0 \leq r , s \leq 1$} and,
in this case, the numbers $r$ and $s$ can be viewed as a discount factor for the prices.
For \mbox{$1 < r$} or \mbox{$1 < s$ } they are a price mark-up.

\begin{definition}
\label{def:local_eq}
Suppose an economy \mbox{$E = (N , X , \: v_{i} , i \in N)$} is given
and let \mbox{$0 \: \leq \: r , s$}.
A \mbox{$ ( r , s ) $}-local equilibrium 
\mbox{$( f , p )$} is a pair where $f$ is a partial allocation of the
items to the agents and $p$ is a price vector that satisfy the following three conditions:
\begin{enumerate}
\item \label{unalloc} 
for any \mbox{$j \in X$} such that \mbox{$f(j) = unallocated$} one has
\mbox{$p_{j} = 0$},
\item \label{individual_rat} {\bf r-Individual Rationality} 
for any \mbox{$i \in N$} one has
\begin{equation} \label{eq:nosale}
v_{i}(S^{f}_{i}) \: \geq \: r \, \sum_{j \in S^{f}_{i}} \: p_{j}, 
\end{equation}
\item \label{out_stability} {\bf s-Outward Stability} 
for any \mbox{$i \in N$} and any \mbox{$A \subseteq X$} such that
\mbox{$A \cap S^{f}_{i} = \emptyset$} one has
\begin{equation} \label{eq:nobuy}
s \,  v_{i}(A \mid S^{f}_{i}) \: \leq \: \sum_{j \in A} \: p_{j}.
\end{equation}
\end{enumerate}
We shall refer to a \mbox{$ ( q , q ) $}-local equilibrium simply as a $q$-local equilibrium.
\end{definition}

We shall be mainly interested in $q$-local equilibria with \mbox{$0 < q \leq 1$}.
Note that {\bf Outward Stability} is formulated for a bundle $A$, not for a single item.
One could have considered a property, stronger than {\bf Individual Rationality}, that would
require that no agent is ready to relinquish even part of his/her bundle at the given discounted 
prices:
\begin{definition} \label{def:strong_equilibrium}
A \mbox{$ ( r , s ) $}-local equilibrium that satisfies, for any \mbox{$i \in N$} 
and any nonempty \mbox{$A \subseteq S_{i}^{f}$} the following
\begin{equation} \label{eq:strong-rationality}
{\bf Strong \ Individual \ Rationality} \ 
v_{i}( A \mid S_{i}^{f} - A) \, \geq \, s \, \sum_{j \in A} p_{j}.
\end{equation}
will be called a strong  \mbox{$ ( r , s )$}-local equilibrium.
\end{definition}
{\bf Individual Rationality} follows from {\bf Strong Individual Rationality} 
by taking \mbox{$A \, = \,$} $S_{i}^{f}$.

Two remarks are in order:
\begin{itemize}
\item we use the term {\em equilibrium} according to usage in Economics, and not according
to its use in Physics. In this paper at least, we do not describe a dynamics that leads to the 
equilibria of Definition~\ref{def:local_eq}.
\item the quality parameters $r$ and $s$ do not measure the quality of the allocation,
they measure the tightness of the fit between allocation and prices. 
Example~\ref{ex:first} below shows that
one can find local equilibria that differ only in the prices, not in the allocation and have different
qualities even though the allocation is the same.
\end{itemize}

It is easy to see that 
\begin{enumerate}
\item any allocation, with any prices, provides a $0$-local equilibrium,
\item if \mbox{$r \: \leq \: r'$} and \mbox{$s \: \leq \: s'$}any 
\mbox{$ ( r' , s' )$}-local equilibrium is a \mbox{$ ( r , s ) $}-local equilibrium, 
\item any conditional equilibrium, and in particular any Walrasian equilibrium 
is a $1$-local equilibrium, and
\item any $1$-local equilibrium is a conditional equilibrium.
\end{enumerate}

Our first example shows that the allocation of a $1$-local equilibrium need not be optimal.
Here and in the forthcoming examples we shall limit ourselves to the study of 
$q$-local equilibria, i.e., in the case \mbox{$r = s$}.
\begin{example} \label{ex:first}
Let \mbox{$X = \{ a , b \}$} and \mbox{$N = \{1 , 2 \}$}.
Both agents are unit-value: \mbox{$v_{1}(a) \: = \: $} \mbox{$v_{2}(b) \: = \:$} $4$, 
\mbox{$v_{1}(b) \: = \:$} \mbox{$v_{2}(a) \: = \:$} $3$ and 
\mbox{$v_{i}(ab) \: = \:$} $4$ for any $i$.
\end{example}
Let $f$ be the sub-optimal allocation which allocates $a$ to agent $2$ and $b$ to agent $1$
and let \mbox{$p_{a} \: = \:$} \mbox{$p_{b} \: = \:$} $2$.
The pair \mbox{$(f , p )$} is a $1$-local equilibrium 
that is not a Walrasian equilibrium.
With the price vector \mbox{$p_{a} \, = \,$} \mbox{$p_{b} \, = \,$} $\frac{1}{2}$
the allocation $f$ provides a $\frac{1}{2}$-local equilibrium.

The question of the existence of high quality local equilibria is postponed to
Section~\ref{sec:local_existence} and we shall now show 
that any \mbox{$ ( r , s ) $}-local equilibrium
provides a \mbox{$(1 + \frac{1}{ r s })$}-approximation of the optimal fractional allocation.

\section{First local social welfare theorem} \label{sec:first}
The first social welfare theorem says that the (partial) allocation $f$ 
of any Walrasian equilibrium has maximal social value:
\mbox{$val(f) \: = \:$} \mbox{$M$}.
We shall show a similar result for \mbox{$ ( r , s ) $}-local equilibria.
Any \mbox{$ ( r , s )$}-local equilibrium, not necessarily strong, is a 
\mbox{$1 + \frac{1}{r s}$}-approximation, i.e., is at least 
\mbox{$\frac{ r s} { 1 + r s }$} efficient.
The strength of this result is that no assumption on the valuations 
of the agents is necessary.
\begin{theorem}[First local social welfare theorem] \label{the:first}
In any exchange economy, if \mbox{$(f , p)$} is a 
\mbox{$ ( r , s ) $}-local equilibrium with \mbox{$r , s > 0$}, then,
for any fractional allocation $x$, one has
\mbox{$val(x) \: \leq$} \mbox{$ ( 1 + \frac{1} { r s } ) \, val(f)$} 
and therefore, for any partial allocation $g$,
\mbox{$val(g) \: \leq \:  ( 1 + \frac{ 1} { r s } )  \, val(f)$}.
\end{theorem}
\begin{proof}
By Lemma~\ref{the:technical} in Appendix~\ref{sec:technical}, 
taking \mbox{$a = \frac{1}{s}$} and \mbox{$b = \frac{1}{r}$}.
\end{proof}

One sees that any $1$-local equilibrium provides a $2$-approximation 
of the fractional optimum.
Theorem~\ref{the:first} therefore improves on Proposition 1 
in~\cite{Fu_Kleinberg_Lavi:EC12}: a conditional equilibrium always provide a 
$2$-approximation of the fractional optimum, not only of the integral optimum.
Theorem~\ref{the:first} implies that the existence of a \mbox{$ ( r , s ) $}-local equilibrium 
implies that the integral gap is at most \mbox{$1 + \frac{1}{r s}$}.
For high quality equilibria the integral gap approaches $1$.

Two examples will now suggest that Theorem~\ref{the:first} cannot be significantly improved.
Our first example shows that, when \mbox{$r = s = 1$} the number 
\mbox{$2 = 1 + \frac{1}{r s}$} cannot be improved upon.
\begin{example}
Suppose \mbox{$X = \{ a , b \}$}, \mbox{$N = \{ 1 , 2 \}$},
\mbox{$v_{1}(a) = 2$}, \mbox{$v_{1}(b) = 1$}, \mbox{$v_{2}(a) =$} $1$,
\mbox{$v_{2}(b) =$} $2$ and \mbox{$v_{1}(ab) =$} \mbox{$v_{2}(ab) =$} $2$.
Both agents are additive with a budget constraint.
The allocation $g$ that gives $a$ to $1$ and $b$ to $2$ has value $4$, 
and with the price vector $(1.5 , 1.5)$ provides a Walrasian equilibrium.
The allocation $g$ is therefore a fractional optimum, by Theorem~\ref{the:Walras}.
The allocation $f$ that gives $b$ to $1$ and $a$ to $2$ with price vector $(1 , 1)$ 
is a $1$-local equilibrium of value $2$, a $2$-approximation.
\end{example}

Our second example will show that Theorem~\ref{the:first} cannot be significantly 
improved upon for low quality equilibria.
\begin{example} \label{ex:smallq}
Consider a single item and two agents. Agent $1$ values the item at $1$ and agent $2$ 
values it at \mbox{$\epsilon \, > \,$} $0$.
The allocation of the item to agent $2$ with a price of $\sqrt{\epsilon}$ is a
$\sqrt{\epsilon}$-local equilibrium.
Theorem~\ref{the:first} claims the allocation is a $1 + \frac{1}{\epsilon}$-approximation 
of the optimal fractional allocation.
The optimal fractional allocation has a value of $1$ and therefore the allocation is, truly,
a $\frac{1}{\epsilon}$-approximation.
For $\epsilon$ close to $0$ Theorem~\ref{the:first} cannot be significantly improved.
\end{example}

\section{Existence} \label{sec:local_existence}
Does every exchange economy possess a $1$-local equilibrium?
We shall discuss two examples.
Our first example is presented in~\cite{GulStacc:99} as an example of an economy 
without a Walrasian equilibrium.

\begin{example} \label{ex:no_Wal}
Let \mbox{$X = \{ a , b , c \}$} and \mbox{$N = \{1 , 2 \}$}.
The two agents have the same valuation: \mbox{$v_{1} \: = \:$} \mbox{$v_{2}$}.
This valuation is symmetric: its gives a zero value to any bundle of less than $2$ items,
a value of $3$ to any bundle of two elements and a value of $4$ to the set $X$.
\end{example}
The optimal fractional allocation gives $1 / 4$ of each of the three bundles of two elements
to each agent: every agent gets, on the whole, $3 / 4$ of a bundle and each item is part
of four bundles, in equal parts. Its value is $4.5$.
The values given to the dual variables by the Dual Linear Program 
(see Section~\ref{sec:dual} of the Appendix)
are \mbox{$\pi_{1} \: = \:$} \mbox{$\pi_{2} \: = \:$} $0$ and
\mbox{$p_{a} \: = \:$} \mbox{$p_{b} \: = \:$} \mbox{$p_{c} \: = \:$} $1.5$.
An optimal allocation gives all three items to any one of the agents and has value $4$.

Let \mbox{$( f , p )$} be a $q$-local equilibrium, with \mbox{$q > 0$}.
One easily sees that all three items must be allocated by $f$: if all three items are unallocated
then both agents violate Outside Stability, therefore there is an agent who is allocated 
some items. 
If he is allocated one or two items this agent 
will violate Outside Stability if there is an item unallocated.
We conclude that all three items are allocated.
If one agent is allocated two items and the second agent is allocated one item, the price
of this last item must be zero by Individual Rationality of the second agent, but then the first
agent violates Outside Stability.
We conclude that in any $q$-local equilibrium with \mbox{$q > 0$} all three items must be 
allocated to the same agent, as in the optimal allocation: in any $q$-local equilibrium 
\mbox{$(f , p )$} with strictly positive $q$ the allocation $f$ is the optimal allocation.
Let us study such $q$-local equilibria.
We have a $q$-local equilibrium if and only if the following inequations are satisfied.
\[
q ( p_{a} + p_{b} + p_{c} ) \: \leq \: 4 \ , \ p_{a} + p_{b} + p_{c} \: \geq \: 4 \, q \ , \ 
p_{a} + p_{b} \: \geq \: 3 \, q \ , \ p_{b} + p_{c} \: \geq 3 \, q \ , 
\ p_{a} + p_{c} \: \geq \: 3 \, q.
\]
Those imply \mbox{$9 q \leq 2 (p_{a} + p_{b} + p_{c} )$} and 
\mbox{$p_{a} + p_{b} + p_{c} \leq 4 / q$}.
We conclude that \mbox{$9 q \leq 8 / q$} and 
\mbox{$q \, \leq \, \frac{2 \, \sqrt2} { 3 }$}.
There is no $q$-local equilibrium for \mbox{$q > \frac{2 \, \sqrt2} { 3 }$}.

But fixing 
\[
p_{a} \, = \, p_{b} \, = \, p_{c} \, = \, \sqrt{2}
\]
and allocating all three items to agent $1$ is a $\frac{2 \, \sqrt2} { 3 }$-local equilibrium,
since all inequations above are satisfied.
We conclude that the best quality attainable is \mbox{$q \, = \, \frac{2 \, \sqrt2} { 3 }$}: 
there is no local equilibrium of quality $1$, and no Walrasian equilibrium.
Note also that Theorem~\ref{the:first} shows  that the social value of the optimal fractional
allocation, which we have seen to be 4.5, is less or equal to \mbox{$(1 + 9 / 8) 4$}.
In other terms, it implies that $val ( f )$ is at least \mbox{$\frac{36}{17}$}, whereas it is, in fact, $4$.

An exchange economy has many different $q$-local equilibria.
Our next example enables us to consider the question: which of those will be attained?
or which of those is the best?
One can think of two general answers.
First, since the prices are driving the market, one can expect the price structure to determine
the allocation that fits the prices.
But one could also expect the market activity to generate an allocation of high social
value and the prices be determined by the allocation.
The question needs further research.

\begin{example} \label{ex:no_eq}
Let \mbox{$X = \{ a , b , c \}$} and \mbox{$N = \{1 , 2 , 3\}$}.
Let $v_{1}(ab)$, $v_{2}(bc)$, $v_{3}(ca)$ and $v_{i}(abc)$ for any $i$ be equal to $1$
and let the values of all other bundles, for any $i$, be equal to $0$.
\end{example}

We are interested in exploring the $q$-local equilibria of this economy.
Let us, first, consider equilibria in which all items have the same price, a reasonable property
since the economy of Example~\ref{ex:no_eq} is unchanged under a permutation of the items.
One may indeed notice that in the fractional optimum, of value $\frac{3}{2}$, the prices of the 
different items are all equal to \mbox{$p \, = \,$} $\frac{1}{2}$.
Let \mbox{$p \, = \,$} \mbox{$p_{a} \, = \,$} \mbox{$p_{b} \, = \,$} \mbox{$p_{c}$}.

In any such local equilibrium of strictly positive quality, $p$ must be strictly positive, 
no item is unallocated, and no agent is allocated a single item 
or a pair of items that he or she values at $0$.
We conclude that in such an equilibrium all items are allocated to a single agent,
which, by the way, provides a social optimum.
Without loss of generality, let us assume that all items are allocated to agent $1$.
The constraints on the price $p$ and the quality $q$ are:
\[
1 \, \geq \, 3 q p \ , \  1 \, \leq \, \frac{2 p}{q}.
\]
We conclude that the highest quality that can be attained by a local equilibrium of this type 
is $\frac{\sqrt{2}}{\sqrt{3}}$. 
Such quality is attained iff \mbox{$p \, = \,$} $\frac{1}{\sqrt{6}}$ .
Note that such a price is less than $\frac{1}{2}$, the price suggested 
by the fractional optimum.
For any inferior quality \mbox{$q \, \leq \,$} $\frac{\sqrt{2}}{\sqrt{3}}$ any price $p$ in
the interval \mbox{$[ \frac{q}{2} , \frac{1}{3 q} ]$} will provide a $q$-local equilibrium.
For \mbox{$ p \, = \,$} $\frac{1}{2}$ the best quality that can be obtained is $\frac{2}{3}$.

We have characterized all $q$-local equilibria in which the three prices are equal and we have 
seen that there is no such $1$-local equilibrium.
Let us now look for $1$-local equilibria.
Any $1$-local equilibrium has, by Theorem~\ref{the:first}, 
a social value of at least $\frac{3}{4}$.
Therefore it has value $1$ and its allocation is optimal.
Without loss of generality, we shall assume that agent $1$ receives $a$ and $b$.
If item $c$ is unallocated or is allocated to one of agents $2$ or $3$ we must have 
\mbox{$p_{c} \, = \, $} $0$ by Individual Rationality if it is allocated or because it is 
unallocated,  \mbox{$1 \, \geq \, $} \mbox{$p_{a} + p_{b}$} by Individual Rationality of
agent $1$, \mbox{$1 \, \leq \,$} $p_{b} + p_{c}$ and \mbox{$1 \, \leq \,$} $p_{a} + p_{c}$, 
by Outside Stability for agents $2$ and $3$.
This is impossible. 
We conclude that any $1$-local equilibrium allocates all three items to the same agent.
The constraints are:
\[
1 \, \geq \, p_{a} + p_{b} + p_{c} \  , \  1 \, \leq \, p_{b} + p_{c} \  , \  
1 \, \leq \, p_{a} + p_{c}.
\]
There is a unique solution: \mbox{$p_{a} \, = \,$} \mbox{$p_{b} \, = \,$} $0$ and
\mbox{$p_{c} \, = \,$} $1$.
If we consider that the economy should attain the optimal allocation in which all items
are allocated to agent $1$ and ask what are the item prices that, with such an allocation,
provide a local equilibrium of the highest quality we find that we can obtain the highest
quality, $1$, with surprising prices: 
\begin{itemize}
\item the items $a$ and $b$ have a zero price notwithstanding
the fact they are valued by agent $1$, and 
\item agent $1$ is ready to pay a high price for an item, $c$, that is useless to him.
\end{itemize}
The explanation may be that item $c$ is of interest to both agents $2$ and $3$ whereas
items $a$ and $b$ are each of interest to one other agent only.

In the exchange economy of Example~\ref{ex:no_eq}, should we expect an
invisible hand to drive the market to an optimal integral solution and 
to prices forming a high quality local equilibrium, 
or should we expect this invisible hand to drive the prices of the different items to be equal?

Note that the allocation that allocates item $a$ to agent $2$, $b$ to agent $3$ and $c$ to
agent $1$ has value $0$. One can check that it is a local optimum 
as defined in Section~\ref{sec:local_optimum}.
Theorem~\ref{the:first} implies that all $q$-local equilibria based on this allocation
have quality \mbox{$q = 0$}.

Note that, with equal prices and \mbox{$p \, = \,$} $\frac{1}{2}$ 
an optimal integral solution such as giving $\{ ab \}$ to agent $1$ 
and letting $c$ be unallocated satisfies all but one condition to be a Walrasian equilibrium: 
every agent gets one of its preferred bundles at the posted prices, 
but $c$ stays unallocated while its price is not zero, and therefore this is not even a local
equilibrium.
Should we expect to see unallocated items with positive prices?

%%Discuss the equilibrium properties of local equilibria: if (f , p) is a local equilibrium
%% and g is an allocation that is close to f, then there are bilateral trades at the prices p
%% that lead from g to f, and no other trades are possible (??)

\section{Bounded complementarity} \label{sec:bounded}
The notion of a valuation that exhibits only limited complementarity, an $a$-submodular 
valuation, \mbox{$1 \leq a$}, was proposed in~\cite{LLN:GEB}.
Submodular valuations are exactly the $1$-submodular valuations and for \mbox{$1 < a$},
$a$-submodular valuations are only approximately submodular.
I wish to propose the thesis that most real life valuations have low complementarity and
that most of the properties of exchange economies
of submodular agents degrade gracefully with the parameter $a$
when $a$-submodular economies are considered.

We shall use $v_{W}$  to denote the marginal valuation defined by 
\mbox{$v_{W}(A) \, = \,$} \mbox{$v(A \cup W ) - v(W)$} for any disjoint bundles $W$, $A$.
The following definition appears in~\cite{LLN:GEB}.
\begin{definition} \label{def:asub}
Let \mbox{$a \, \geq \, 1$}. A valuation $v$ is said to be {\em $a$-submodular}
iff for any \mbox{$W , A \subseteq X$}, \mbox{$W \cap A \, = \,$} $\emptyset$,
and for any \mbox{$x \in X - W - A$}
\[
v_{W}(A \cup \{ x \} ) \: \leq \: v_{W}(A) + a \, v_{W}(x).
\]
\end{definition}

An obvious example of a valuation that is $a$-submodular for {\em no} $a$ is a valuation
that values at $0$ each of two items separately but values them at a strictly positive value
together. 
But note that, since the set of items $X$ is finite, in any valuation $v$ 
that is $a$-submodular for no $a$, there must be a bundle \mbox{$W \subseteq X$} and
an item \mbox{$x \in X - W$} such that \mbox{$v(W \cup \{ x \} ) =$}
\mbox{$v(W)$}.
One sees that, given any valuation $v$, there is a valuation $v'$ that is arbitrarily close
to $v$ that is $a$-submodular for some (typically large) $a$.
Valuations that exhibit unbounded complementarity have been considered 
in the literature (see, e.g., \cite{LCS:JACM} ) but, in real economies, 
it seems that complementarity is bounded.
Note that if each of the two items above have value $1$ and not $0$ and the pair has value
$4$, a considerable complementarity, the valuation is still a $3$-submodular valuation.
The properties of $a$-submodular valuations are studied in Appendix~\ref{app:a-sub}.

\section{Local equilibria in $a$-submodular economies} \label{sec:opt_equil}
In an exchange economy in which every agent's valuation is $a$-submodular, 
requirements similar to those of Definition~\ref{def:local_eq}, but dealing with a single item
only and not sets of items, guarantee a strong $q$-local equilibrium for 
\mbox{$q = \frac{1}{a}$}.

\begin{theorem} \label{the:submod}
In any $a$-submodular exchange economy, if $f$ is a partial allocation 
and $p$ is a price vector that satisfy
\begin{enumerate}
\item \label{unalloc_sub}
for any \mbox{$j \in X$} such that \mbox{$f(j) = unallocated$} one has
\mbox{$p(j) = 0$},
\item \label{no_sub}
for any \mbox{$i , k \in N$}, \mbox{$i \neq k$} and any \mbox{$j \in S_{i}^{f}$} 
one has
\[
v_{k}(j \mid S_{k}^{f} ) \: \leq \: p_{j} \: \leq \: v_{i}(j \mid S_{i}^{f} - \{ j \} )
\]
\end{enumerate}
then the pair $(f , p)$ is a strong $\frac{1}{a}$-local equilibrium.
\end{theorem}
\begin{proof}
Let \mbox{$i \in N$} and \mbox{$A \subseteq S_{i}^{f}$}.
By Lemma~\ref{the:main_bounded} in Appendix~\ref{app:a-sub}, 
\[
\sum_{j \in A} v_{i}(j \mid S_{i}^{f} - \{ j \} ) \: \leq \: a \, v_{i} (A \mid S_{i}^{f} - A).
\]
Therefore, by our assumption, Equation~(\ref{eq:strong-rationality}) is satisfied with
\mbox{$q = \frac{1}{a}$}.
For the Outward Stability property, let \mbox{$i \in N$} and \mbox{$A \subseteq X$},
\mbox{$A \cap S_{i}^{f} = \emptyset$}.
By Lemma~\ref{the:main_bounded},  and our assumption
\[
v_{i}(A \mid S_{i}^{f}) \: \leq \: a \, \sum_{j \in A} v_{i}( j  \mid S_{i}^{f} ) \: \leq \:
a \, \sum_{j \in A} p_{j}.
\]
\end{proof}

An example will show that Theorem~\ref{the:submod} cannot be improved significantly.
\begin{example} \label{ex:asubmod}
Suppose two items and two agents.
Agent $1$ values any of the items to $1$ and both items to $1 + a$ (\mbox{$a \geq 1$}).
His valuation is $a$-submodular.
Agent $2$ has an additive valuation: each item is valued at $1$ 
and the whole set of two items at $2$.
\end{example}
The allocation that gives both items to agent $2$ with prices $1$ to each item satisfies
the assumptions of Theorem~\ref{the:submod}.
It is a $\frac{2}{1+ a}$-local equilibrium since
\mbox{$2 \, \leq \, \frac{1 + a}{2} \, 2$}, \mbox{$1 \, \geq \, \frac{2}{1+ a} \, 1$} and
\mbox{$1 + a \, \leq \, \frac{1 + a}{2} \, 2$}.
When $a$ is large,  \mbox{$\frac{2}{1+ a}$} is of the same order as \mbox{$\frac{1}{a}$}.

\section{Local optima} \label{sec:local_optimum}
We shall now recall the definition of a local optimum as
presented in~\cite{Babaioff_Endowment:EC18}.
It formalizes the notion of an allocation that is Pareto-optimal under simple transfers of 
single items.
We shall then show that, in an $a$-submodular exchange economy, any local optimum
can be associated with a price vector to form a strong $\frac{1}{a}$-local equilibrium.
In~\cite{Babaioff_Endowment:EC18} the authors show that, in a submodular economy,
every local optimum is a $2$-approximation of the fractional optimum.
Theorem~\ref{the:equilibrium} generalizes this result.

In an exchange economy, agents trade items and they can trade in many different,
sometimes complex, patterns involving a number of agents.
But bilateral trades, i.e., trades between two agents seem to be most prevalent.
It even seems that, typically, bilateral trades consist of one agent selling a bundle to
another agent: one agent delivers a bundle and receives money, the other agent gives money
and receives a bundle.
Most prevalent seems to be the transfer of a single item, in exchange for money,
from an agent to another one.
If we limit ourselves to the consideration of such simple bilateral actions, we expect,
at the long end, to find the economy in a situation in which no such bilateral trade can
be profitable to both the seller and the buyer.
Such situations are natural candidates for allocations that are part of 
some kind of equilibrium.
Such a situation has been termed a {\em local optimum} in~\cite{Babaioff_Endowment:EC18}.
Note that no prices are involved here.

\begin{definition} \label{def:local_optimum}
An total allocation \mbox{$f : X \longrightarrow N $} is said to be a
{\em local optimum} iff for any distinct agents \mbox{$i , k \in N$}, \mbox{$i \neq k$}
and for any item \mbox{$j \in S_{i}^{f}$} allocated to agent $i$, 
one has:
\begin{equation} \label{eq:Pareto}
v_{i}(S_{i}^{f} - \{ j \} ) + v_{k}(S_{k}^{f} \cup \{ j \}) \: \leq \: 
v_{i}(S_{i}^{f}) + v_{k}(S_{k}^{f}) ,
\end{equation}
equivalently \mbox{$v_{i}(j \mid S_{i}^{f} - \{ j \}) \: \geq \: v_{k}(j \mid S_{k}^{f})$}.
\end{definition}

If the allocation of items is a local optimum, in a secondary market only complex trades 
will be performed: transfers of bundles, exchanges, or trades involving more than two agents.

Note that any allocation that maximizes social value, i.e., any global optimum,
is a local optimum.
Therefore any exchange economy possesses a local optimum.

Any local optimum defines in a natural way, for each item, 
a set of prices: prices that support the allocation of the item to the agent it is
allocated to in a second price auction.

\begin{definition} \label{def:suitable}
Let $f$ be a local optimum and let \mbox{$j \in X$}.
The agent $f(j)$ is the agent to whom $j$ is allocated and therefore
\mbox{$v_{f(j)} ( j \mid S_{f(j)}^{f} - \{ j \} ) \: \geq \:$}
\mbox{$v_{i} ( j \mid S_{i}^{f} ) $} for any agent \mbox{$i \neq f(j)$}.
We say that any number $\alpha$ such that, for any agent \mbox{$i \neq f(j)$}
\[
v_{i} ( j \mid S_{i}^{f} ) \, \leq \, \alpha \, \leq \, v_{f(j)} ( j \mid S_{f(j)}^{f} - \{ j \} ) 
\]
is a {\em suitable} price for item $j$ given the local optimum $f$ and that
any price vector $p$ such that $p_{j}$ is a suitable price for every item $j$ given $f$ is
a {\em supporting} price vector for $f$.
\end{definition}

Caution: the term {\em supporting} has a different meaning in
~\cite{Dobzinski_Nisan_Schapira:STOC05, Fu_Kleinberg_Lavi:EC12}.
The following is obvious.
\begin{lemma} \label{the:supporting}
Every local optimum admits a supporting price vector.
\end{lemma}

In any exchange economy one can obtain a local optimum by starting
from any allocation and executing a sequence of moves in which a single item
is transferred from an agent to another one, 
if this move strictly benefits the social value.
The procedure must terminate in a local optimum.
The complexity of finding a local optimum has been studied 
in~\cite{Babaioff_Endowment:EC18} and its 
communication complexity has been studied 
in~\cite{Babi_Dobz_Nisan:search_arxiv}.
It follows from results there that the sequence of moves above can be of an exponential length
even for submodular economies.

\begin{theorem} \label{the:equilibrium}
In an $a$-submodular exchange economy, if $f$ is a local optimum
and $p$ is a supporting price vector then the pair \mbox{$( f , p )$} is a strong
$\frac{1}{a}$-local equilibrium
and \mbox{$val(x) \, \leq \,$} \mbox{$( 1 + a^{2} ) \, val(f)$}
for any fractional allocation $x$.
Therefore any local optimum is a $1 + a^{2}$-approximation of the fractional optimum.
\end{theorem}
\begin{proof}
By Theorems~\ref{the:submod} and~\ref{the:first}.
\end{proof}
Note that \mbox{$( f , p )$} is not claimed to be an $a$-endowed equilibrium as defined 
in~\cite{Babaioff_Endowment:EC18}: the latter is a global optimum 
whereas a local equilibrium is a local optimum.
It follows from Theorem~\ref{the:equilibrium} that in an $a$-submodular exchange economy 
the integral gap is at most $1 + a^{2}$, but Corollary~\ref{the:co_gap} will provide a better
bound.

The following corollary is a second welfare theorem for $1$-local equilibria, i.e., conditional
equilibria. It strenghtens Proposition 3 and Corollary 1 of~\cite{Fu_Kleinberg_Lavi:EC12} 
very significantly: 
it applies to any local optimum, not only to a welfare maximizing allocation, and to 
$a$-submodular economies (for any $a$) not only to submodular economies.
\begin{corollary} \label{the:opt-equi}
In an $a$-submodular economy, if $f$ is a local optimum, 
then there is a price vector $p$ such that \mbox{$(f , p)$} is a $1$-local equilibrium.
\end{corollary}
\begin{proof}
By Lemma~\ref{the:supporting}.
\end{proof}

Note that, in Example~\ref{ex:no_eq},  the valuations are not $a$-submodular for any $a$. 
Theorem~\ref{the:equilibrium} cannot be applied to the local optimum of value $0$ described
there.

Note that, in Example~\ref{ex:asubmod}, the allocation of both items to agent $2$ is
a local optimum that is a $\frac{1 + a}{2}$-approximation of the fractional optimum.
It is easy to see that there is no worse local optimum and therefore every local optimum
is a $\frac{1 + a}{2}$-approximation of the fractional optimum.
Theorem~\ref{the:equilibrium} claims only that every local optimum is a
$1 + a^{2}$-approximation.
I do not know of an economy for which the bound in Theorem~\ref{the:equilibrium} is sharp.

Note also that the $\frac{1}{a}$ bound on the quality of the local equilibrium holds for any
set of supporting prices, but some vectors of supporting prices may provide
local equilibria of better quality than others.

\section{Greedy allocation in economies with bounded complementarity} \label{sec:greedy}
In Section~\ref{sec:local_optimum} we described how, in an $a$-submodular exchange economy,
a local optimum defines a local equilibrium,
we shall now describe a different way to obtain a local equilibrium.
The family of greedy allocation algorithms introduced in~\cite{LLN:GEB} was claimed
there to provide a $1 + a$-approximation of the integral optimum when all the agents'
valuations are $a$-submodular.
Such algorithms require only polynomial time. 
An improved result has been presented orally at~\cite{Lehmann:Dagstuhl}, to the effect that,
for \mbox{$a \, = \,$} $1$, they provide a $2$-approximation of the {\em fractional} optimum.
We shall now show that greedy algorithms provide a $\frac{1}{a}$-local equilibrium, 
and a $1 + a$-approximation of the fractional optimum.
This is better than the $1 + a^{2}$-approximation guaranteed by 
Theorem~\ref{the:equilibrium}.
A greedy allocation can be implemented by a sequence of single-item auctions, 
auctioning the items separately. 

A greedy allocation consists in the choice of a total ordering of the items of $X$:
\mbox{$j_{1} , \dots , j_{m}$}.
An iterative process then allocates the items one by one in the order chosen: an item is
allocated to the agent for which it has the highest marginal value.
At stage $0$ we set \mbox{$S_{i}^{0} = \emptyset$} for any \mbox{$i \in N$}.
At stage $k$, for \mbox{$k = 1 , \ldots , m$}  we choose an agent $i_{k}$ such that
\mbox{$v_{i_{k}}(j_{k} \mid S_{i_{k}}^{k - 1} ) \geq v_{l}(j_{k} \mid S_{l}^{k - 1})$} 
for any agent $l$ and set
\mbox{$S_{i_{k}}^{k} = S_{i_{k}}^{k - 1} \cup \{ j_{k} \} $} and 
\mbox{$S_{l}^{k} = S_{l}^{k - 1}$} for any agent $l$, \mbox{$l \neq i_{k}$}.
The resulting allocation $f$ is defined by \mbox{$f(j_{k}) = i_{k}$} for any $k$.
The procedure may be used to define a price for each of the items.
The price $p_{k}$ of item $k$ is fixed, at the time $k$ is allocated, 
at any value less or equal to its marginal value for the agent it is allocated to 
and larger or equal to its marginal value for any of the other agents.
The price of item $k$, once fixed, is never modified.

\begin{theorem} \label{the:greedy}
In an $a$-submodular exchange economy, any greedy allocation algorithm 
results in an allocation that, with the prices defined just above, provides a 
$( 1 , \frac{1}{a} ) $-local equilibrium.
The allocation obtained, $f$, is a $1+a$-approximation, i.e.,
\mbox{$val(x) \, \leq \,$} \mbox{$(1 + a) \, val(f)$} for any fractional allocation $x$.
\end{theorem}
\begin{proof}
The allocation provided is a total allocation, therefore condition~\ref{unalloc} of
Definition~\ref{def:local_eq} is satisfied.

We show, by induction on $k$, that, for any $k$, \mbox{$0 \leq k \leq m$} 
and for any agent $i$:
\begin{equation} \label{eq:greedy_indiv}
v_{i} (S_{i}^{k} ) \, \geq \, \sum_{l \in S_{i}^{k}} p_{l}.
\end{equation}
First, for any agent $i$:
\[
v_{i} (S_{i}^{0}) \, = \, v_{i} ( \emptyset) \, = \, 0 \, \geq \, \sum_{l \in \emptyset} p_{l}.
\]
Let $i$ be the agent to which item $k + 1$ is allocated.
For any agent $d$ different from $i$, 
\mbox{$S_{d}^{k + 1} \, = \, $} \mbox{$S_{d}^{k}$} and Equation~(\ref{eq:greedy_indiv})
holds for $d$ and $k + 1$ by the induction hypothesis.
Since the price of $k + 1$, $p_{k + 1}$ is less or equal to item $k + 1$'s marginal value 
for $i$
\[
v_{i}(S_{i}^{k + 1} ) \, = \, v_{i}(S_{i}^{k}) + v_{i}(k + 1 \mid S_{i}^{k}) \, \geq \,
\sum_{l \in S_{i}^{k}} p_{l} + p_{k + 1} \, = \, \sum_{l \in S_{i}^{k + 1}} p_{l}.
\]
We conclude that the final allocation satisfies Individual Rationality with $r = 1$.

We now want to show that for any agent $i$, any stage $k$ and any bundle $A$ 
{\em of already allocated items}, 
\mbox{$A \subseteq \bigcup_{l \in N} S_{l}^{k}$},
\mbox{$A \cap S_{i}^{k} = \emptyset$}  we have
\begin{equation} \label{eq:greedy_outward}
v_{i}(A \mid S_{i}^{k} ) \, \leq \, a \, \sum_{l \in A} p_{l}.
\end{equation}
After the allocation of item $k$, we only need to check the two cases below.
\begin{itemize}
\item For the agent $i$ to whom $k$ has been allocated.
A bundle $A$ such that \mbox{$A \cap S_{i}^{k + 1} \, = \,$}$\emptyset$ 
does not include $k$.
We have, by Lemma~\ref{the:ASk}
\[
v_{i} (A \mid S_{i}^{k} \cup \{ k \} ) \, \leq \, 
a \, \sum_{x \in A} v_{i}(x \mid S_{i}^{r(x)}) \, \leq \,
a \, \sum_{x \in A} p_{x}
\]
where $r(x)$ is the stage at which item $x$ has been allocated.
\item For any other agent $j$ for any $A$ that includes $k$.
By the induction hypothesis, Lemma~\ref{the:bounded} and the choice of $p_{k}$
\[
v_{j} ( A' \cup \{ k \} \mid S_{j} ) \: = \: 
v_{j} ( A' \mid S_{j} ) + v_{j} (k \mid S_{j} \cup A' ) \: \leq \:
\]
\[
a \, \sum_{l \in A'} p_{l} + a \, v_{j} ( k \mid S_{j} ) \: \leq \:
a \, \sum_{l \in A'} p_{l} + a \, p_{k} \, = \,
a \, \sum_{l \in A} p_{l}.
\]
\end{itemize}
We have shown that the greedy allocation, together with the prices defined by the
greedy process satisfy the conditions of Theorem~\ref{the:technical} 
with \mbox{$b \, = \,$} $1$.
Our claims now follow from the theorem.
\end{proof}
Note that the local equilibrium obtained is not, in general, a strong local equilibrium
(with discount factor $r$ equal to $1$)
since the marginal value of item $k$ for the agent to whom it has been allocated 
is different in the final allocation from what it was at the time $k$ was allocated and
$p_{k}$ was set. This marginal value may have decreased and 
may now be smaller than $p_{k}$.

The following follows immediately from Theorem~\ref{the:greedy} and generalizes
a result of~\cite{Feige_subadditive:SIAM} for submodular economies.
\begin{corollary} \label{the:co_gap}
In an $a$-submodular economy the integral gap is at most $1 + a$.
\end{corollary}

The following example shows that a greedy allocation does not always provide 
a strong local optimum even in submodular economies.
\begin{example} \label{ex:no_Pareto}
Consider two items $a$, $b$ and two agents $1$, $2$.
Let \mbox{$v_{1}(a) =$} \mbox{$ v_{1}(b) =$} $5$, \mbox{$v_{1}(ab) =$} $7$ and 
\mbox{$v_{2}(a) =$} $4$, \mbox{$v_{2}(b) =$} $1$ and \mbox{$v_{2}(ab) =$} $5$.
\end{example}
Both valuations are submodular.
If, in a greedy allocation, $a$ is allocated before $b$, agent $1$ is allocated $a$ and $b$.
But, then, \mbox{$v_{1}(a \mid b) =$} \mbox{$ 2 < 4 =$} \mbox{$ v_{2}(a)$}.

\section{Substitutes economies} \label{sec:substitutes}
In an exchange economy in which all agents have a substitutes valuation, one may
pinpoint exactly which of the $1$-local equilibria are Walrasian:
if no agent is interested in exchanging, at the posted prices, a single item allocated to him 
for an item not in his possession.
The proof is short and relies of the Single Improvement property of substitutes valuations 
(\cite{GulStacc:99}).
\begin{theorem} \label{the:substitutes}
In an exchange economy in which all agents have substitutes valuations, \mbox{$( f , p )$} 
is a $1$-local equilibrium such that for any 
\mbox{$i \in N$}, any \mbox{$j \in S_{i}^{f}$} and any
\mbox{$k \in X - S_{i}^{f}$} one has
\mbox{$v_{i}(S_{i}^{f}) - v_{i}(S_{i}^{f} - j + k) \: \geq \:$}
\mbox{$p_{j} - p_{k}$}, iff \mbox{$( f , p )$} is a Walrasian equilibrium.
\end{theorem}
\begin{proof}
The {\em if} part is obvious and does not need the substitutes assumption.
For the {\em only if} part assume $v_{i}$ is substitutes for any agent $i$ and that $p$ is 
a price vector.
The valuation \mbox{$u_{i}(A) \: = \:$} \mbox{$v_{i}(A) - \sum_{j \in A} p_{j}$} is also
substitutes. 
The assumptions ensure that, for every \mbox{$i \in N$}, 
\mbox{$u_{i}(S_{i}^{f} ) \: \geq \:$} \mbox{$u_{i}(A)$} for any 
\mbox{$A \subseteq X$} such that the size of the symmetric difference 
\mbox{$S_{i}^{f} \Delta A$} is less or equal to $2$.
The single improvement condition shown to be equivalent 
to the substitutes property in Theorem~1 of~\cite{GulStacc:99} 
implies that $S_{i}^{f}$ maximizes $u_{i}$ over all subsets of $X$.
We conclude that \mbox{$(f , p )$} is a Walrasian equilibrium.
\end{proof}

\section{Summary and open questions} \label{sec:open}
This paper proposes the notion of a \mbox{$ ( r , s ) $}-local equilibrium to understand the role of item
prices in discrete exchange economies.
It focuses on such economies in which all agents have an $a$-submodular valuation.
Two different processes that build local equilibria have been put in evidence.
The first one is based on simple bilateral trades and seems close 
to the way real markets function. 
It provides a local optimum with a price vector that is defined by the allocation.
This allocation is a $1 + a^{2}$-approximation of the fractional optimum.
The second one is greedy allocation with historical prices, prices corresponding to the moment
the item has been allocated. 
It may resemble the birth of a market accommodating more and more items.
This allocation is a $1 + a$-approximation of the fractional optimum.
The discrepancy in the quality of approximations requires further study.
Do random greedy allocations really provide higher social value than the local optima obtained
from random initial allocations by sequences of simple bilateral trades?
How do simple bilateral trades perform on initial allocations that are already the result
of a greedy process?

The following questions require for further research.
Can the $1 + a^{2}$-approximation for any local optimum be improved?
At the moment no $a$-submodular economy with a local optimum that is only
a $1 + a^{2}$-approximation is known, for \mbox{$a \, > \, 1$}.
Can one prove a better approximation result for a restricted class of local optima,
e.g., strong local optima or Pareto optimal allocations under all bilateral trades?
How prevalent can the absence of a $1$-local equilibrium be?
Can the optimal allocation always upport prices that exhibit 
the highest quality local equilibrium?
Could it be that most typical economies have a $1$-local equilibrium?
Are local equilibria typically stable, i.e., does a small change in valuations or in prices
bring only a small change in the local equilibrium?
Some high quality local equilibria, as in Example~\ref{ex:no_eq}, seem surprising.
Do all $1$-local equilibria have economic significance?
What is the dynamics of the revelation of such equilibrium prices?

The results presented in this paper do not depend on the number of agents or items.
There is, I think, a general feeling that a better equilibrium can be reached in an a large 
economy, i.e., an economy in which a large number of agents actively participate.
Could it be that the approximation obtained by any local optimum 
in which a large number of agents are allocated a non-empty bundle is better than the one
promised in Theorem~\ref{the:equilibrium}?

As noticed in~\cite{LLN:GEB} maximizing social welfare in a discrete economy is a problem
of maximizing a function over a matroid. 
One should consider our results from this point of view too.
Is the notion of a local maximum interesting there? 
The notion of a function close to submodular?

Can the notions of local optimum and local equilibrium be of use in the study of markets of
divisible goods?

Finally, one should look for dynamics that lead to local equilibria of high quality.

\section{Conclusion: the role of prices} \label{sec:conclusion}
The view presented in this paper is that markets attain a local equilibrium through the advent
of suitable item prices. 
The role of prices in this process is significantly different from their role according to the view
that markets attain a Walrasian equilibrium. 
According to this last view, prices, through a {\em tatonnement} process, converge towards
equilibrium prices that are the best possible: if agents accept those 
prices and trade, at those prices, to improve their individual welfare, every agent will find
himself in the best possible situation.
No profitable trade is prevented by the equilibrium prices.
Apart from the convergence towards equilibrium prices, we do not expect prices to vary,
and any divergence from the equilibrium prices can only hamper progress towards equilibrium.
If an invisible hand would reveal equilibrium prices from the start, 
convergence towards equilibrium would only be sped up.

In the local equilibrium view of prices, prices have a different role.
They play the traditional role of guiding the market towards a (local) equilibrium,
but, once such a local equilibrium is attained, those prices can prevent trades that 
would improve the social welfare.
Consider a $q$-local equilibrium (\mbox{$q < 1$}) in which agent $1$ holds an item 
he values at \mbox{$x$}, but that agent $2$ values at \mbox{$y > x$}.
Note that the allocation is not a local optimum and that a trade would improve the social
welfare.
If the price $p$ of the item is greater than $y$ or less than $x$, no trade at price $p$ can take
place, even though it could happen at another price.
In such a situation local equilibrium prices may have a negative effect: 
they can prevent an increase in social welfare resulting from trade. 
There, a change in prices may enable profitable trades that were impossible previously.
We expect that changes in the revealed prices can help the market to move from a local
equilibrium to another local equilibrium of higher social value.
Such price modifications may also change the quality of a local equilibrium, 
but the forces behind such process are still unclear.

\section{Acknowledgments} \label{sec:Ack}
Exchanges with Ron Lavi, Moshe Babaioff, Shahar Dobzinski and Noam Nisan are gratefully
acknowledged.

\bibliographystyle{plain}
%%\bibliography{../../../my}

\appendix

\section{Walrasian equilibria} \label{sec:Walras_app}
In a Walrasian equilibrium $(f , p)$, every agent \mbox{$i \in N$}, at the given prices $p$,
prefers his allocated bundle $S_{i}^{f}$ to any other bundle.
For an agent \mbox{$i \in N$}, its {\em utility} for bundle \mbox{$A \subseteq X$}
is defined as \mbox{$u_{i}(A) \: = \:$} \mbox{$v_{i}(A) - \sum_{j \in A} \: p_{j}$}.
\begin{definition}
\label{def:Walras_eq}
Suppose an economy \mbox{$E = (N , X , \: v_{i} , i \in N)$} is given.
A Walrasian equilibrium \mbox{$( f , p )$} is a pair where $f$ is a partial allocation of the
items to the agents and $p$ is a price vector that satisfy the following conditions:
\begin{enumerate}
\item \label{unalloc_w}
for any \mbox{$j \in X$} such that \mbox{$f(j) = unallocated$} one has
\mbox{$p_{j} = 0$},
\item \label{noregret}
for any \mbox{$i \in N$} and any \mbox{$D \subseteq X$} one has
\begin{equation} \label{eq:noregret}
u_{i}(S^{f}_{i}) \: \geq \: u_{i}(D).
\end{equation}
\end{enumerate}
\end{definition}

In a Walrasian equilibrium every agent is allocated the bundle he prefers amongst all
possible bundles, if only he considers he cannot have any influence on the prices.
A Walrasian equilibrium is the best of all possible situations for each and every agent.

The basic properties of Walrasian equilibria are described in Theorem~\ref{the:Walras}. 
They summarize the two theorems of welfare economics 
and the results of~\cite{BikhMamer:Equ} in an original manner.

\begin{theorem} \label{the:Walras}
\begin{itemize}
\item \label{frac_opt}
If $(f , p)$ is a Walrasian equilibrium then \mbox{$val(f) = M_{F}$},
i.e., $f$ is a fractional optimum,
\item \label{exists_p}
if $f$ is an allocation and \mbox{$val(f) = M_{F}$}, then there exists a price vector $p$
such that $(f , p)$ is a Walrasian equilibrium,
\item \label{same_p}
if $(f , p)$ is a Walrasian equilibrium and $g$ is an allocation such that 
\mbox{$val(g) = M_{F}$}, then $(g , p)$ is a Walrasian equilibrium.
\end{itemize}
\end{theorem}
\begin{proof}
First, under the assumptions, for any fractional allocation $x_{i}^{D}$:
\[
val(x) \: = \: \sum_{i \in N} \sum_{D \subseteq X} x_{i}^{D} v_{i}(D) \: \leq \:
\sum_{i \in N} \sum_{D \subseteq X} x_{i}^{D} v_{i}(S_{i}^{f}) \: = \:
\]
\[
\sum_{i \in N} v_{i}(S_{i}^{f}) \sum_{D \subseteq X} x_{i}^{D} \: \leq \:
\sum_{i \in N} v_{i}(S_{i}^{f}) \: = \: val(f).
\]

Secondly, under the assumptions, $f$ is the solution to the linear program \textbf{LP}.
The dual program \textbf{DLP} is described in Appendix~\ref{sec:dual}.
We shall take the variables $p_{j}$, $j \in X$ of the dual for prices in the equilibrium.
It follows from general results on linear programming, that for any \mbox{$i \in N$}, 
\mbox{$D \subseteq X$} 
\begin{itemize}
\item \mbox{$\pi_{i} \geq v_{i}(D) - \sum_{j \in D} \:  p_{j}$} and
\item if \mbox{$x_{i}^{D} > 0$} then
\mbox{$\pi_{i} \: = \: v_{i}(D) - \sum_{j \in D} \: p_{j}$}.
\end{itemize}
But \mbox{$x_{i}^{S_{i}^{f}} = 1$} and therefore
\[
v_{i}(S_{i}^{f}) - \sum_{j \in S_{i}^{f}} p_{j} \: = \: \pi_{i} \: \geq \: 
v_{i}(D) - \sum_{j \in D} \: p_{j}.
\]

For the third part of our claim, we shall show that, for any \mbox{$i \in N$},
one has 
\[
v_{i}(S_{i}^{g}) - \sum_{j \in S_{i}^{g} } \: p_{j} \: = \: 
v_{i}(S_{i}^{f}) - \sum_{j \in S_{i}^{f} } \: p_{j}.
\]
Since $(f , p)$ is a Walrasian equilibrium, we know that, for any \mbox{$i \in N$},
\[
v_{i}(S_{i}^{f}) - \sum_{j \in S_{i}^{f} } \: p_{j} \: \geq \: 
v_{i}(S_{i}^{g}) - \sum_{j \in S_{i}^{g} } \: p_{j}.
\]
But, since \mbox{$val(g) = val(f)$} and then by the fact that if \mbox{$f(j) = unalloc$}
one has \mbox{$p_{j} = 0$}
\[
\sum_{i \in N} ( v_{i}(S_{i}^{g}) - \sum_{j \in S_{i}^{g}} p_{j} ) \: = \: 
\sum_{i \in N} v_{i}(S_{i}^{g}) - \sum_{i \in N} \sum_{j \in S_{i}^{g}} p_{j} \: \geq \:
\sum_{i \in N} v_{i}(S_{i}^{f}) - \sum_{j \in X} p_{j} \: = \:
\]
\[
\sum_{i \in N} v_{i}(S_{i}^{f}) - \sum_{i \in N} \sum_{j \in S_{i}^{f} } p_{j} \: = \:
\sum_{i \in N} ( v_{i} (S_{i}^{f} ) - \sum_{j \in S_{i}^{f}} p_{j}  ).
\]
\end{proof}

\section{A technical lemma} \label{sec:technical}
The following shows sufficient conditions for an allocation to be a $1 + a \, b$-approximation
(\mbox{$a , b \geq 0$}) of the fractional optimal allocation.
It is used to prove Theorem~\ref{the:first}.
Note that in an exchange economy the agents' valuations are assumed to satisfy both
Free disposal and Normalization but that Normalization is not used 
in Lemma~\ref{the:technical}.
Its strength is that no further assumption is made on the agents' valuations.
\begin{lemma} \label{the:technical}
In any exchange economy, let \mbox{$a \, , b \: \geq \: 0$} and assume that 
$f$ is a partial allocation and $p$ is a price vector that satisfy the following two conditions:
\begin{enumerate}
\item \label{nosale_app}
for any \mbox{$i \in N$} 
\begin{equation} \label{eq:nosale_app}
b \, v_{i}(S^{f}_{i}) \: \geq \: \sum_{j \in S_{i}^{f}} \: p_{j},
\end{equation}
\item \label{nobuy_app}
for any \mbox{$i \in N$} and any \mbox{$A \subseteq X$} such that
\mbox{$A \cap S^{f}_{i} = \emptyset$} one has
\begin{equation} \label{eq:nobuy_app}
v_{i}(A \mid S^{f}_{i}) \: \leq \: a \, \sum_{j \in A} \: p_{j},
\end{equation}
\end{enumerate}
then, for any fractional allocation $x$:
\begin{equation} \label{eq:technical}
val(x) \: \leq \: (1 + a \, b ) \: val(f) + a \, \sum_{j \in X , f(j) = unalloc} p_{j}.
\end{equation}
\end{lemma}
\begin{proof}
By definition,
\mbox{$val(x) = \sum_{i \in N} \sum_{D \subseteq X} x^{i}_{D} v_{i}(D)$}.
By the free disposal assumption, then:
\[
val(x) \: \leq \:
\sum_{i \in N} \sum_{D \subseteq X} x^{i}_{D} v_{i}(D \cup S^{f}_{i}) \: = \:
\sum_{i \in N} \sum_{D \subseteq X} x^{i}_{D} v_{i}(S^{f}_{i}) +
\sum_{i \in N} \sum_{D \subseteq X} x^{i}_{D} v_{i}(D - S^{f}_{i} \mid S^{f}_{i}).
\]
First, by Equation~(\ref{eq:con1LPR})
\[
\sum_{i \in N} \sum_{D \subseteq X} x^{i}_{D} v_{i}(S^{f}_{i}) \: = \:
\sum_{i \in N} v_{i}(S^{f}_{i}) \sum_{D \subseteq X} x^{i}_{D} \: \leq \:
\sum_{i \in N} v_{i}(S^{f}_{i}) \: = \:
val(f).
\]
Then, by Equation~(\ref{eq:nobuy_app}), then Equation~(\ref{eq:con2LPR}) and finally by
Equation~(\ref{eq:nosale_app}) we have
\[
\sum_{i \in N} \sum_{D \subseteq X} 
x_{i}^{D} v_{i}(D - S^{f}_{i} \mid S^{f}_{i}) \: \leq \:
\sum_{i \in N} \sum_{D \subseteq X} a \, x_{i}^{D} \sum_{j \in D - S^{f}_{i}} p_{j} 
\: \leq \: \sum_{j \in X} \sum_{i \in N} \sum_{D \subseteq X , j \in D} 
a \, x_{i}^{D} p_{j} \: = \:
\]
\[
a \, \sum_{j \in X} p_{j} \sum_{i \in N} \sum_{D \subseteq X , j \in D} x_{i}^{D} \: \leq \:
a \, \sum_{j \in X} p_{j} \: = \: 
a \, \sum_{i \in N} \sum_{j \in S_{i}^{f}} p_{j} + 
a \, \sum_{j \in X , f(j) = unalloc} p_{j} \: \leq \:
\]
\[
a \, b \, \sum_{i \in N} v_{i}(S_{i}^{f}) + a \, \sum_{j \in X , f(j) = unalloc} p_{j} \: = \: 
a \, b \, val(f) + a \, \sum_{j \in X , f(j) = unalloc} p_{j}.
\]
We conclude that \mbox{$val(x) \leq (1 + a \, b ) \: val(f) + a \, \sum_{j \in X , f(j) = unalloc} p_{j}$}.
\end{proof}

\section{Dual linear program} \label{sec:dual}
The dual of \textbf{LP} will be described now.
\begin{center}
\textbf{Dual Linear Program (DLP):}
\end{center}
\nopagebreak Minimize
\begin{equation}
\label{eq:objDLP}
\sum_{j \in X} \: p_{j} + \sum_{i \in N} \: \pi_{i}
\end{equation}
under the constraints
\begin{equation}
\label{eq:posDLP}
p_{j} \geq 0 , \pi_{i} \geq 0 {\rm \ for \ all \ } j \in X , i \in N ,
{\rm \ and \ }
\end{equation}
\begin{equation}
\label{eq:conDLP}
\sum_{j \in D} \: p_{j} + \pi_{i} \geq  v_{i}(D) , {\rm \ for \ all \ }
D \subseteq X {\rm \ and \ } i \in N.
\end{equation}

\section{Properties of $a$-submodular valuations} \label{app:a-sub}
Some of the results below have been claimed without proof in~\cite{LLN:GEB}.
\begin{lemma} \label{the:bounded}
Let $v$ be $a$-submodular.
For any \mbox{$S, T, A \subseteq X$}, such that
\mbox{$S \subseteq T$} and \mbox{$A \cap T = \emptyset$} we have
\[
v(A \mid T) \, \leq \, a \, v(A \mid S).
\]
\end{lemma}
\begin{proof}
Let $v$, $S$, $T$ and $A$ be as in the assumptions.
We shall prove our claim by induction on the size of $A$.
If \mbox{$A \, = \, \emptyset$} the claim is obvious.
For the induction step, let \mbox{$x \in X - A - T$}.
\[
v(A \cup \{ x \} \mid T ) \, = \, v(A \mid T) + v(x \mid A \cup T).
\]
By the induction hypothesis \mbox{$v(A \mid T) \, \leq \,$} \mbox{$ a \, v(A \mid S)$}.
Let \mbox{$B \, = \,$} \mbox{$A \cup T$} and \mbox{$C \, = \,$} \mbox{$A \cup S$}.
We have \mbox{$C \subseteq B$} and
\[
v(x \mid B ) \, = \, v( C \cup (B - C) \cup \{ x \} ) - v(C \cup (B - C)) \, = \,
v_{C}((B - C) \cup \{ x \} ) - v_{C}(B - C) \, \leq \,
\]
\[
v_{C} ( B - C ) + a \, v_{C}(x)  - v_{C}(B - C) \, = \,
a \, v(x \mid C).
\]
We conclude that
\[ 
v(A \cup \{ x \} \mid T ) \, \leq \, a \, v(A \mid S) + a \, v(x \mid A \cup S) \, = \, 
a \, v(A \cup \{ x \} \mid S ).
\]
\end{proof}

%%\begin{corollary} \label{the:AUB}
%%Let $v$ be $a$-submodular.
%%If \mbox{$A , B \subseteq X$} and \mbox{$A \cap B = \emptyset$},
%%then 
%%\[
%%v(A \cup B) \, \leq \, v(A) + a \, v(B).
%%\]
%%\end{corollary}
%%\begin{proof}
%%By Lemma~\ref{the:bounded} with \mbox{$T = B$} and \mbox{$S = \emptyset$},
%%\[ 
%%v(B \cup A) - v(A) \, = \, v(B \mid A) \, \leq \, a \, v(B).
%%\]
%%\end{proof}

\begin{lemma} \label{the:ASk}
Let $v$ be $a$-submodular. Let $A$, $B$ be disjoint bundles and for each \mbox{$x \in A$}
let \mbox{$B_{x} \subseteq B$}.
Then one has:
\[
v(A \mid B) \, \leq \, a \, \sum_{x \in A} v( x \mid B_{x} ).
\]
\end{lemma}
\begin{proof}
By induction on the size of $A$. For \mbox{$A = \emptyset$} the claim is obvious.
For the induction step, let \mbox{$y \in X - A - B$} and \mbox{$B_{y} \subseteq B$}.
By the induction hypothesis and Lemma~\ref{the:bounded}.
\[
v(A \cup \{ y \} \mid B) \, = \, v(A \mid B) + v(y \mid A \cup B) \, \leq \,
\]
\[
a \, \sum_{x \in A} v( x \mid B_{x}) + a \, v( y \mid B_{y} ) \, = \,
a \, \sum_{x \in A + y} v(x \mid B_{x} ).
\]
\end{proof}

The following result will be instrumental in Section~\ref{sec:opt_equil}.
\begin{lemma} \label{the:main_bounded}
Let the valuation $v$ be $a$-submodular.
\begin{enumerate}
\item
For any \mbox{$A \subseteq S \subseteq X$}.
\[
a \: v(A \mid S - A ) \: \geq \: \sum_{j \in A} v( j \mid S - \{ j \} ) 
\]
and 
\item for any \mbox{$A , S \subseteq X$} such that \mbox{$A \cap S \, = \, \emptyset$},
\[
v( A \mid S ) \: \leq \: a \: \sum_{j \in A} v( j  \mid S ).
\]
\end{enumerate}
\end{lemma}
\begin{proof}
Let \mbox{$A \, = \, \{ j_{1} , j_{2} , \ldots , j_{k} \}$} and let 
\mbox{$A_{0} = \emptyset$}, \mbox{$A_{1} = \{ j_{1} \}$}, 
\mbox{$A_{i} = \{ j_{1} , \ldots , j_{ i } \}$} for \mbox{$i \geq 2$}.
We have 
\[ 
v(A \mid S - A ) \, = \,
\sum_{i = 0}^{k - 1} v(j_{i + 1} \mid ( S - A ) \cup A_{i} ).
\]
By Lemma~\ref{the:bounded} we have
\mbox{$v(j_{i + 1} \mid S - \{ j_{i + 1 } \} ) \, \leq \,$}
\mbox{$a \, v( j_{i + 1} \mid ( S - A ) \cup A_{i} )$}
and this proves our first claim.
Now, by Definition~\ref{def:asub},
\[
v( A \mid S ) \, \leq \, v( A_{k} \mid S ) + a \, v(j_{k} \mid S ) \, \leq \, 
\]
\[
v(A_{k - 1} \mid S ) + 
a \, v(j_{k - 1}\mid S ) + a \, v(j_{k} \mid S )
\, \leq \, \ldots \, \leq \,  a \, \sum_{i = 1}^{k} v(j_{i} \mid S ).
\]
\end{proof}

\section{Quasi-Walrasian equilibria} \label{sec:quasi-Walras}
A quasi-Walrasian equilibrium of quality $q$, \mbox{$0 \leq q \leq 1$} consists of a partial
allocation and a price vector such that every agent gets from his bundle, at the given prices, 
a utility that is at least the utility he would get from any other bundle discounted by $q$.
\begin{definition} \label{def:quasi-Walras}
Suppose an economy \mbox{$E = (N , X , \: v_{i} , i \in N)$} is given
and let \mbox{$0 \: \leq \: q \: \leq \: 1$}.
A $q$-quasi-Walrasian equilibrium 
\mbox{$( f , p )$} is a pair where $f$ is a partial allocation of the
items to the agents and $p$ is a price vector that satisfy the following two conditions:
\begin{enumerate}
\item \label{unallocW} 
for any \mbox{$j \in X$} such that \mbox{$f(j) = unallocated$} one has
\mbox{$p_{j} = 0$},
\item \label{W-cond} 
for any \mbox{$i \in N$} and for any \mbox{$A \subseteq X$} one has
\begin{equation} \label{eq:nosaleW}
v_{i}(S^{f}_{i}) - \sum_{j \in S^{f}_{i}} p_{j}\: \geq \: q \, ( v_{i}(A) - 
\sum_{j \in A} \: p_{j} ).
\end{equation}
\end{enumerate}
\end{definition}
Note that the notion of a Walrasian equilibrium in the $\alpha$-endowed valuations studied 
in~\cite{Babaioff_Endowment:EC18} is incomparable with that of a 
$\frac{1}{\alpha}$-quasi-Walrasian equilibrium.
Clearly a pair \mbox{$(f , p )$} is a $1$-quasi-Walrasian equilibrium iff it is a Walrasian
equilibrium and any allocation together with zero prices 
provides a $0$-quasi-Walrasian equilibrium.

The value of the allocation of a $q$-quasi-Walrasian equilibrium is at least $q$ times the social
optimum.
\begin{theorem} [First social quasi welfare theorem] \label{the:first_quasi}
If \mbox{$(f , p)$} is a $q$-quasi-Walrasian equilibrium, then 
\mbox{$val(f) \geq q \: val(g)$} for any partial allocation $g$.
\end{theorem}
Note that the approximation $q$ here is better than the \mbox{$\frac{q^{2}}{1 + q^{2}}$}
of Theorem~\ref{the:first}, but that the comparison is, here, with the integral social optimum, 
not with the fractional, higher, optimum.
\begin{proof}
Let $g$ be the social optimum.
We have, by condition~\ref{unallocW} and then condition~\ref{W-cond} 
of Definition~\ref{def:quasi-Walras}:
\[
val(f) \: = \: \sum_{i \in N} v_{i}(S^{f}_{i})  \:  = \: 
\sum_{i \in N} ( v_{i} ( S^{f}_{i}) - \sum_{j \in S^{f}_{i}} p_{j} ) + \sum_{j \in X} p_{j}
\: \geq \: 
\]
\[
q \: \sum_{i \in N} (v_{i}(S^{g}_{i} ) - \sum_{j \in S^{g}_{i}} p_{j} ) + 
\sum_{j \in X} p_{j} \: = \:
q \: \sum_{i \in N} v_{i}(S^{g}_{i} ) - q \: \sum_{i \in N} \sum_{j \in S^{g}_{i}} p_{j} + 
\sum_{j \in X} p_{j} \: \geq \: q \: val(g).
\]
\end{proof}

\begin{lemma} \label{le:W-local}
Any $q$-quasi-Walrasian equilibrium is a $q$-local-equilibrium.
\end{lemma}
\begin{proof}
Let \mbox{$(f , p)$} be a $q$-quasi-Walrasian equilibrium.
Let us show that the three conditions of Definition~\ref{def:local_eq} are satisfied.
Condition~\ref{unalloc} is explicitly satisfied by Definition~\ref{def:quasi-Walras}.
For Individual Rationality, notice that, for any \mbox{$i \in N$},
\[
v_{i}(S^{f}_{i}) - \sum_{j \in S^{f}_{i}} p_{j} \: \geq \: 
q \: ( v_{i}(\emptyset) - 0 ) \: = \: 0.
\]
Equation~(\ref{eq:nosale}) is satisfied even with a parameter $q$ equal to $1$.
For Outward Stability note that 
\[
v_{i}(S^{f}_{i}) - \sum_{j \in S^{f}_{i}} p_{j} \: \geq \: q \: (v_{i}(S^{f}_{i} \cup A) - \sum_{j \in S^{f}_{i} \cup A} p_{j} )
\] 
and therefore
\[
\sum_{j \in A} p_{j} \: \geq \: v_{i}(S^{f}_{i} \cup A) - v_{i}(S^{f}_{i}) \: = \: 
v_{i}( A \mid S^{f}_{i}).
\]
\end{proof}

Let us now consider the $q$-quasi-Walrasian equilibria of Example~\ref{ex:no_Wal}.
Both agents have the same valuation $v$.
We shall show that there are no such equilibria for any \mbox{$q > 0$}.
First, assume that agent $1$ receives an empty bundle in a $q$-quasi-Walrasian equilibrium.
His utility is $0$ and therefore, comparing with receiving two of the three items, we see that
\[
0 \: \geq \: q \: (3 - p_{1} - p_{2}) , 0 \: \geq \: q \: (3 - p_{1} - p_{3}) , 0 \: \geq \: 
q \: (3 - p_{2} - p_{3})
\]
and therefore, if \mbox{$q > 0$}, \mbox{$p_{1} + p_{2} + p_{3} \geq 4.5$}.
Note that the parameter $q$ has disappeared from the inequality.
If agent $2$ is allocated the whole bundle of three items, he must prefer this 
to the empty bundle and we must have:
\[
4 - p_{1} - p_{2} - p_{3} \: \geq \: q \: 0 \: = \: 0
\]
which is impossible.
But, if an item, say item $3$ is unallocated, we must have \mbox{$p_{3} = 0$} and
\mbox{$p_{1} + p_{2} \geq 4.5$}.
Allocating the bundle \mbox{$(1 , 2)$} to agent $2$ would imply 
\[
3 - p_{1} - p_{2} \: \geq \: q \: 0 \: = \: 0
\]
which is impossible. 
We conclude that at least two items must be unallocated, but this can be shown similarly
to imply that no item is allocated, which is clearly impossible.
We have shown that there is no $q$-quasi-Walrasian equilibrium for \mbox{$q > 0$} in which
some agent receives an empty bundle.

Suppose now that agent $1$ is allocated a single item, say item $3$.
We must have
\[
0 - p_{3} \: \geq \: q \: 0 = 0
\]
and therefore \mbox{$p_{3} = 0$}.
Since agent $1$ prefers item $3$ to all three items we have:
\[ 
0 - p_{3} = 0 \: \geq \: q \: ( 4 - p_{1} - p_{2})
\]
and \mbox{$p_{1} + p_{2} \geq 4$}.
In such a situation, agent $2$ cannot be allocated the pair \mbox{$(1 , 2 )$}, nor can he be allocated any single item.

We conclude that, if \mbox{$q > 0$}, in any $q$-quasi-Walrasian eqilibrium each of the two
agents must be allocated at least two of the three items, which is impossible.
The only $q$-quasi-Walrasian equilibria have \mbox{$q = 0$}.

\end{document}